\documentclass[sn-mathphys]{sn-jnl}

\usepackage{amsmath}
\usepackage{amsthm}
\usepackage{amssymb}
\usepackage{mathtools}
\usepackage{braket}
\usepackage{physics}
\usepackage{float}
\usepackage{graphicx}
\usepackage[bottom]{footmisc}
\usepackage{qcircuit_new}
\usepackage{url}
\usepackage{siunitx}
\usepackage{bbm}
\usepackage{xcolor}
\usepackage[justification=centering]{caption}

\usepackage{soul}

\DeclareMathOperator{\E}{\mathbb{E}}
\newtheorem{lemma}{Lemma}

\theoremstyle{thmstyleone}%
\newtheorem{theorem}{Theorem}
\theoremstyle{thmstyletwo}%
\theoremstyle{thmstylethree}%
\newtheorem{definition}{Definition}%

\jyear{2021}%
\begin{document}
\title[{Quantum Approximation of Normalized Schatten Norms and its Applications}]{Quantum Approximation of Normalized Schatten Norms and Applications to Learning}

\author[1]{\fnm{Yiyou} \sur{Chen}}\email{gerry99@ucla.edu}
\author[2]{\fnm{Hideyuki} \sur{Miyahara}}\email{hmiyahara512@gmail.com}
\author[3,4,5,6,7]{\fnm{Louis-S.} \sur{Bouchard}}\email{louis.bouchard@gmail.com}
\author[2,4]{\fnm{Vwani} \sur{Roychowdhury}}\email{vwani@g.ucla.edu}
\affil[1]{\orgdiv{Department of Computer Science}, \orgname{University of California}, \orgaddress{ \street{404 Westwood Plaza}, \city{Los Angeles}, \postcode{90095}, \state{California}, \country{USA}}}

\affil[2]{\orgdiv{Department of Electrical and Computer Engineering}, \orgname{University of California}, \orgaddress{\street{420 Westwood Plaza}, \city{Los Angeles}, \postcode{90095}, \state{California}, \country{USA}}}

\affil[3]{\orgdiv{Department of Chemistry and Biochemistry}, \orgname{University of California}, \orgaddress{\street{607 Charles E Young Dr E}, \city{Los Angeles}, \postcode{90095}, \state{California}, \country{USA}}}

\affil[4]{\orgdiv{Center for Quantum Science and Engineering}, \orgname{University of California}, \orgaddress{\street{Los Angeles}, \city{Los Angeles}, \postcode{90095}, \state{California}, \country{USA}}}

\affil[5]{\orgdiv{California NanoSystems Institute}, \orgname{University of California}, \orgaddress{\street{570 Westwood Plaza building 114}, \city{Los Angeles}, \postcode{90095}, \state{California}, \country{USA}}}

\affil[6]{\orgdiv{Department of Bioengineering}, \orgname{University of California}, \orgaddress{\street{410 Westwood Plaza}, \city{Los Angeles}, \postcode{90095}, \state{California}, \country{USA}}}

\affil[7]{\orgdiv{The Molecular Biology Institute}, \orgname{University of California}, \orgaddress{\street{
611 Charles E. Young Drive East}, \city{Los Angeles}, \postcode{90095}, \state{California}, \country{USA}}}

\abstract{
Efficient measures to determine similarity of quantum states, such as the fidelity metric, have  been widely studied. In this paper, we address the problem of defining a similarity measure for quantum operations that can be \textit{efficiently estimated}.  Given two quantum operations, $U_1$ and $U_2$, represented in their circuit forms, we first develop a quantum sampling circuit to estimate the normalized Schatten 2-norm of their difference ($\norm{U_1-U_2}_{S_2}$) with precision $\epsilon$, using \textit{only one clean qubit and one classical random variable}. We prove a $\text{Poly}(\frac{1}{\epsilon})$ upper bound on the sample complexity, which is independent of the size of the quantum system. 
We then show that such a similarity metric is directly related to a functional definition of similarity of unitary operations using the conventional fidelity metric of quantum states ($\mathcal{F}$): If 
$\norm{U_1-U_2}_{S_2}$ is sufficiently small (e.g. $ \leq \frac{\epsilon}{1+\sqrt{2(1/\delta - 1)}}$) then the fidelity of  states obtained by processing the same randomly and uniformly picked pure state, $\ket{\psi}$, is as high as needed ($\mathcal{F}({U}_1\ket{\psi}, {U}_2\ket{\psi})\geq 1-\epsilon$) with probability exceeding $1-\delta$.
We provide example applications of this efficient similarity metric estimation framework to quantum circuit learning tasks, such as finding the square root of a given unitary operation. 
	}



\maketitle

\section{Introduction}
Recent advances in quantum approximate optimization algorithms (QAOA, \cite{farhi2014quantum,farhi2016quantum}), variational quantum eigensolver (VQE, \cite{Peruzzo_2014}) and the promise of   implementing such algorithms using noisy intermediate-scale quantum (NISQ) devices \cite{Preskill_2018} have rekindled the prospect of a new era in quantum computing.
Researchers have started experimenting with quantum machine learning algorithms such as quantum neural networks (QNN) \cite{KAK1995259,ezhov2000quantum} and quantum circuit learning \cite{Mitarai_2018, panella2011neural,gingrich2004non,schlimgen2021quantum,20092} that are based on variational quantum algorithms \cite{Cerezo_2021, lubasch2020variational}; a recent work has studied the price of ansatz used in such variational methods\cite{https://doi.org/10.48550/arxiv.2102.01759}. 
These algorithms assume a hybrid model  which takes advantage of both classical and quantum computations: loss functions are obtained by summing the outputs of a quantum machine whereas the variational parameters of the model (circuit) are learned using a classical optimizer. 

A critical factor in the formulation of a learning algorithm is the design of its loss functions, which often involves computing a similarity measure between a target objective and the output of the parameterized model (Fig. \ref{vqa}). In VQE, for instance, the objective is to determine the  lowest energy eigenstate of a given Hermitian operator $H$. The learning framework assumes an ansatz comprising a quantum circuit with a fixed topology but where  each gate is parameterized to generate a  candidate eigenstate vector  $\ket{\psi(\xi)}=U(\xi)\ket{\boldsymbol{0}}$, where $U(\xi)$ is the unitary operator determined by the parameters $\xi$. Such a pure state vector $\ket{\psi(\xi)}$ is an eigenstate, if  $H \ket{\psi(\xi)}=\lambda\ket{\psi(\xi)}$, where $\lambda$ is to be minimized when searching for the ground state. Thus, the objective function of VQE can be interpreted as minimizing the cosine similarity between $\ket{\psi(\xi)}$ and $H\ket{\psi(\xi)}$ or the expectation value  $\bra{\psi(\xi)}H\ket{\psi(\xi)}$. This loss term can be physically estimated by performing measurements corresponding to the observables used to define $H$. 
This similarity metric is related to the well known measure of fidelity used for determining similarity of quantum states,  \cite{nielsen_chuang_2010,doi:10.1080/09500349414552171,pedersen2007fidelity}, and has been extensively applied to distinguishing quantum states \cite{200721,article2,20012,10.1145/276698.276708,2015,doi:10.1080/09500340903477756,200910,20082, watrous2018theory}. 

In contrast, consider the problem where one wants to learn a given quantum operation $V$, which is also available when controlled by a clean qubit (see Table \ref{vqackt}). That is, the task is to learn $\xi$ such that $U(\xi) \approx V$. This problem of learning quantum operations is much less studied, in spite of its applications to quantum circuit synthesis \cite{10.1145/3503222.3507739} (where low-depth approximations of quantum circuits are needed) and to  distinguish quantum operations and channels \cite{PhysRevA.82.032339,PhysRevLett.98.230502,PhysRevA.78.012303,4787618,PhysRevLett.103.210501,PhysRevA.81.032329,lu2010optimal,10.5555/2231036.2231045,PhysRevLett.109.020506,6747300,7541701,6687245,article,200911,shimbo2018equivalence,2001,Yang2005DistinguishabilityCI,2005,2006,PhysRevLett.98.100503,Li2014ABO,200722,article4, watrous2018theory, chen2008ancillaassisted}. The main difficulty of such learning tasks is the design of a similarity metric between $U(\xi)$ and $V$ that can be \textit{efficiently estimated}. 
{In Gilchrist et al.'s work  \cite{Gilchrist_2005}, for example, a similarity metric that is blind to input unitary operations was studied and validated, but the estimation of the metric is inefficient because it requires an exponential number of quantum states. 
Other metrics such as the diamond norm \cite{https://doi.org/10.48550/arxiv.quant-ph/9806029,Regula_2021} have been conceptualized to distinguish quantum operations, but they heavily rely on the classical information of the input unitaries such as their eigen-decompositions. }

The Schatten norm, studied and explored from an information theory perspective \cite{P_rez_Garc_a_2006,Ben_Aroya_2008,Hayden_2008}, is another candidate, and the approximation of which has been proven to be DQC1-complete \cite{https://doi.org/10.48550/arxiv.0707.2831,https://doi.org/10.48550/arxiv.1706.09279}. These approximation schemes \cite{https://doi.org/10.48550/arxiv.0707.2831,https://doi.org/10.48550/arxiv.1706.09279}, however, require an exponential classical sample complexity when a clean qubit {(e.g. the control bit of the Hadamard test circuit in Table \ref{vqa})} is provided. Herein {we present a random sampling method, using few samples (e.g. $O(\frac{\log(2/\delta)}{2\epsilon^2})$ sample complexity) and an efficient sampling circuit design (e.g. $O(1)$ in depth), 
to estimate the normalized trace and normalized Schatten 2-norm of any given quantum operation.} 
We then formulate a similarity metric for quantum operations using the notion of fidelity of quantum states, and show how such a metric is closely related to the normalized Schatten 2-norm. As a consequence, one can use the normalized Schatten 2-norm of the difference between $U(\xi)$ and $V$ as a loss function to learn a target quantum operation $V$. 
\begin{figure}[ht]
\centering
\includegraphics[width=8cm]{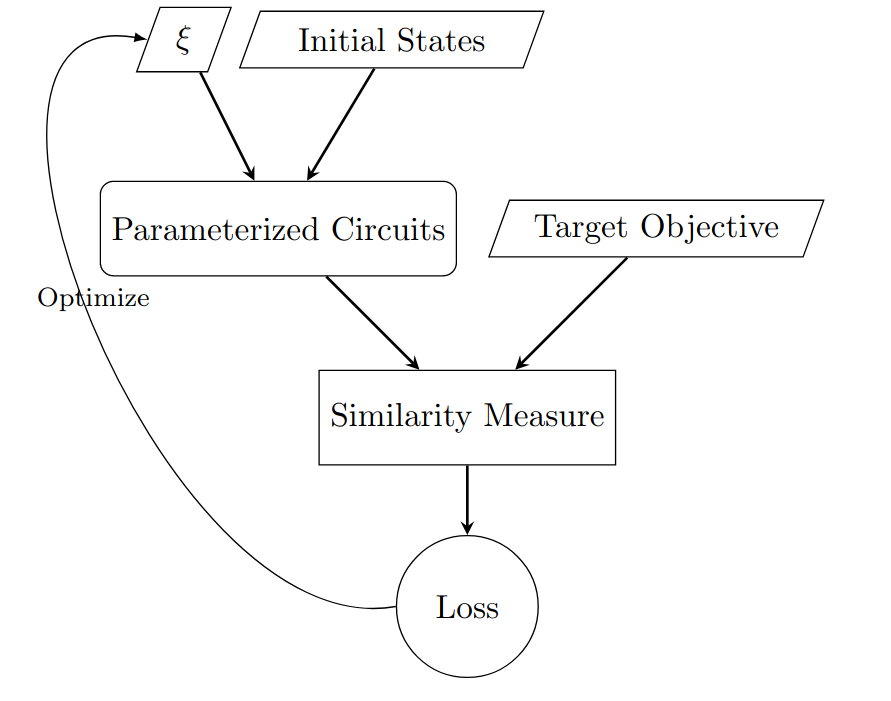}
\caption{A schematic illustration of variational quantum algorithms. In the case of VQE, the initial state is $\ket{\mathbf{0}}$, the parameterized circuit is $U$, the target objective is a Hamiltonian $H$, the similarity measure is the cosine similarity, and the loss term is $\bra{\mathbf{0}}U^\dagger(\xi)H U(\xi)\ket{\mathbf{0}}$.}
\label{vqa}
\end{figure}
\begin{table}[ht]
\[\Qcircuit @C=1em  @R=1em {
&\lstick{\ket{0}} & \gate{H} &\qw&\ctrl{1}&\qw&\gate{H}&\qw&  \meter\\
&\lstick{\ket{\boldsymbol{0}}} & {/}\qw&\gate{{U(\xi)}}&\gate{V}&\qw
&\qw&\qw\\
}\]
\caption{Hadamard test circuit. Given an arbitrary quantum operation $V$ controlled by a clean qubit, the above circuit computes $\Re{\bra{\psi(\xi)}V\ket{\psi(\xi)}}$, where $\ket{\psi(\xi)} = U(\xi)\ket{\boldsymbol{0}}$ and $U(\xi)$ is a quantum circuit parameterized by $\xi$. }\label{vqackt}
\end{table}

The paper is organized as follows. In Section \ref{background}, we provide the background concepts and notations. 
In Section \ref{sampling-ckt}, we present a sampling method to approximate the {normalized} trace of matrices that are unitarily similar to diagonal matrices, as well as to approximate the normalized Schatten 2-norm of arbitrary $N\cross M$ matrices. We prove an upper bound of the sample complexity of such a sampling method. In Section \ref{linear-comb}, we introduce the normalized Schatten 2-norm of mixed quantum operations, which can be estimated efficiently using the sampling method from the previous section and the Hadamard test circuits shown in Tables \ref{vqackt}. We also present an optimized circuit design for the approximation. In Section \ref{similarity}, we relate the normalized Schatten 2-norm to a similarity metric of quantum operations. {Finally in Section \ref{learning}, we present an application of the efficient approximation of the normalized Schatten 2-norm to quantum circuit learning.}
\section{Background and Notations}
\label{background}
{Given an $n$-qubit quantum system, a quantum state is specified by a density matrix $\rho = \sum_{i}p_i\ket{\psi_i}\bra{\psi_i}$, where $\ket{\psi_i} \in \mathbb{CP}^{N-1}$ ($\mathbb{CP}$ denotes the complex projective space) are pure state vectors and $N=2^n$ is the dimension of the Hilbert space representing the quantum system.} 
To deal with the equivalence class on $\mathbb{CP}^{N-1}$ we adopt the convention {of normalization} that a pure state $\ket{\psi}$ is a point on the boundary of a unit ball centered at the origin, i.e. $\ket{\psi}\in \partial B_{\mathbb{C}^N}(0, 1)$, and thus has a unit norm, i.e. $\bra{\psi}\ket{\psi} = 1$. 
We use the term ``quantum operation'' to refer to a  unitary map of a density matrix $\rho \rightarrow U \rho U^\dagger$.
This linear operator is a type of a Liouville space superoperator.
A unitary operator $U\in \mathbb{C}^{N\cross N}$ working on $n$ qubits can generally be decomposed as a product of unitary operators, where each $U_i$ is a unitary operator acting on only a reduced number of qubits.  For example: $U= \Pi_{i=1}^L U_i$, $U_i= U_i^{j,k}\otimes I^{s-\{j,k\}}$ where $s$ is the set of qubits and $U_i^{j,k}\in \mathbb{C}^{4\cross 4}$ acts on the $j^{th}$ and $k^{th}$ qubits. Such few-qubit operations {(e.g. $U_i ^{j,k}$)} are referred to as gates, and a quantum circuit is a visual representation of a sequence of gates used to represent a quantum operation. 
Some well-known quantum gates include 
$$\sigma_x = \begin{pmatrix} 0&1\\1&0
\end{pmatrix}, \quad \sigma_y=\begin{pmatrix}
0&-i\\
	i&0
\end{pmatrix}, \quad  \sigma_z=\begin{pmatrix}
1&0\\
0&-1
\end{pmatrix}, \quad H=\begin{pmatrix}
\frac{1}{\sqrt{2}} & \frac{1}{\sqrt{2}}\\
\frac{1}{\sqrt{2}} & \frac{-1}{\sqrt{2}}
\end{pmatrix}
$$
$$
\text{CNOT}=\begin{pmatrix}
1 & 0 & 0 & 0\\
0 & 1 & 0 & 0\\
0 & 0 & 0 & 1\\
0 & 0 & 1 & 0
\end{pmatrix}, \quad S=\begin{pmatrix}
1&0\\
0&i
\end{pmatrix}
$$
\begin{align*}R_x(\theta)&= e^{-i\frac{\theta}{2}\sigma_x}=\begin{pmatrix}
\cos(\frac{\theta}{2}) & -i\sin(\frac{\theta}{2})\\
-i\sin(\frac{\theta}{2})  &\cos(\frac{\theta}{2})
\end{pmatrix}\\
R_y(\theta)&= e^{-i\frac{\theta}{2}\sigma_y}=\begin{pmatrix}
\cos(\frac{\theta}{2}) & -\sin(\frac{\theta}{2})\\
\sin(\frac{\theta}{2})  &\cos(\frac{\theta}{2})
\end{pmatrix}\\
R_z(\theta)&= e^{-i\frac{\theta}{2}\sigma_z}=\begin{pmatrix}
e^{-i\frac{\theta}{2}} & 0\\
0&e^{i\frac{\theta}{2}}
\end{pmatrix}.
\end{align*}
	
In the simplest instance we can define a parametrized quantum circuit $U(\theta)$ as a circuit with learnable parameters $\theta$ for its rotational gates. The parameters $\theta$ are called variational parameters.

Given an arbitrary quantum unitary $V$ and a state $\ket{\psi}$, the expectation of $V$, $\bra{\psi}V\ket{\psi}$, can be estimated using a method called the Hadamard test \cite{10.1145/1132516.1132579}, {as shown in Table \ref{hadamard1} and Table \ref{hadamard2}. Letting $\text{Pr}(1)$ be the probability of observing $\ket{1}$ from measuring the control qubit in Table \ref{hadamard1} and Table \ref{hadamard2}, then $1-2\Pr(1)$ evaluates $\Re\{\bra{\psi}V\ket{\psi}\}$ and $\Im\{\bra{\psi}V\ket{\psi}\}$, respectively.} The control bit in the Hadamard test is called a clean qubit.
\begin{table}[ht]
\[\Qcircuit @C=1em @R=.6em {
& \lstick{\ket{0}} &\gate{H} &\ctrl{1} &\gate{H} &\qw &  \meter\\
&\lstick{\ket{\psi}}&{/}\qw& \gate{V}& \qw & \qw
}\]
\caption{Hadamard test: $\Re\{\bra{\psi}V\ket{\psi}\}$.}\label{hadamard1}
\end{table}
\begin{table}[ht]
\[\Qcircuit @C=1em @R=.6em {
	& \lstick{\ket{0}} &\gate{H} &\gate{S^\dagger} &\ctrl{1} &\gate{H} &\qw &  \meter\\
	&\lstick{\ket{\psi}}&{/}\qw&\qw &\gate{V}& \qw & \qw
}\]
\caption{Hadamard test: $\Im\{\bra{\psi}V\ket{\psi}\}$.}\label{hadamard2}
\end{table}

In this paper{,} we also consider a generalization of the definition of quantum operations comprising finite linear combinations of unitary quantum operations. We use $\tilde{U}$ to denote such mixed quantum operations: $\tilde{U}=\sum_{\kappa =1}^K \alpha_\kappa U_\kappa$, {where $U_\kappa\in \mathbb{C}^{N\times N}$ are unitary quantum operations} and $\sum_{\kappa=1}^K \abs{\alpha_{\kappa}}\leq 1$. For simplicity, we only consider the case when $K$ is of $O(1)$. 
{In the case when $\alpha_{\kappa}\geq 0$ and $\sum_{\kappa=1}^K\alpha_{\kappa} = 1$, a mixed quantum operation $\sum_{\kappa =1}^K \alpha_\kappa U_\kappa$ represents a quantum computing mixture model that
applies the unitary operation $U_{\kappa}$ to any state $\rho$ with probability $\alpha_{\kappa}$ and yields a density matrix output of $\sum_{\kappa=1}^K \alpha_{\kappa} U_{\kappa}\rho U_{\kappa}^\dagger$.} Such a mixed quantum system is a special case of a larger class, Completely-Positive Trace-Preserving (CPTP) Maps \cite{nielsen_chuang_2010}, and could in principle be used to model quantum errors. In general, one could ask the following question: Given {a mixed quantum} system as an oracle, can one design a variational quantum algorithm to {approximate it to a high accuracy}. This is one of the learning problems {we address} in Section~\ref{learning}.

For any matrix $A\in \mathbb{C}^{N\cross M}$, let $A=W\Sigma T^ \dagger$ be the singular value decomposition (SVD), where $\Sigma = \text{diag}(\sigma_i)$ and $\sigma_i\geq 0$ are the singular values. Note that $AA^\dagger$ is always diagonalizable with eigenvalues $\sigma_i^2$, and $AA^\dagger = W\Sigma T^\dagger T\Sigma^\dagger W^\dagger=W(\Sigma\Sigma^\dagger)W^\dagger=W\hat{\Sigma}W^\dagger $, where $\hat{\Sigma} = \textup{diag}(\sigma_i^2)$. Moreover, let $\ket{w_i}$ be the column vectors of $W$ (also called {left-singular vectors}) under bra-ket notation, $AA^\dagger = \sum_{i=1}^N \sigma_i^2\ket{w_i}\bra{w_i}$. 

A square matrix $A\in \mathbb{C}^{N\cross N}$ is unitarily similar to a diagonal matrix $D$ if $A=WDW^\dagger$ where $W$ is unitary and $D$ is diagonal in $\mathbb{C}^{N\cross N}$. In particular, for any matrix $A\in \mathbb{C}^{N\cross M}$, $AA^\dagger$ is unitarily similar to $\text{diag}(\sigma_i^2)$. Another result we use is that all unitary matrices are also unitarily similar to  diagonal matrices \cite{sakurai_napolitano_2017}.

Given two quantum states' density operators $\rho_1, \rho_2$, the fidelity is customarily defined as 
$\mathcal{F}_\rho(\rho_1,\rho_2)=\left[ \Tr(\sqrt{\sqrt{\rho_1}\rho_2\sqrt{\rho_1}})\right]^2$ {\cite{nielsen_chuang_2010}}. 
Given two pure states $\rho_1=\ket{\psi_1}\bra{\psi_1}$ and $\rho_2=\ket{\psi_2}\bra{\psi_2}$, it can be shown that $\mathcal{F}_\rho(\rho_1,\rho_2)= \abs{\bra{\psi_1}\ket{\psi_2}}^2$. In this paper we work with pure states, and this  simplified version of fidelity between wavefunctions will be used and denoted as $\mathcal{F}(\ket{\psi_1}, \ket{\psi_2})$. In general, fidelity measures how similar two quantum states are, and $\psi_1=\psi_2$ if and only if $\mathcal{F}(\psi_1, \psi_2)=1$.
\begin{definition}
The \textbf{normalized Schatten $p$-norm} \cite{P_rez_Garc_a_2006,Ben_Aroya_2008,Hayden_2008,https://doi.org/10.48550/arxiv.0707.2831,https://doi.org/10.48550/arxiv.1706.09279} for arbitrary matrix $A\in \mathbb{C}^{N\cross M}$ and {$p\in [1,\infty)$} is defined as \begin{align*}\norm{A}_{S_{p}} &= \Big(\frac{\sum_i{{\sigma_i}^p}}{N}\Big)^{\frac{1}{p}}
\end{align*}
where $\sigma_i$ are the singular values of $A$. Note that the Schatten 2-norm is related to the Frobenius norm via $\norm{A}_{S_2}=\frac{\norm{A}_F}{\sqrt{N}}$.
\end{definition}
We can relate the normalized Schatten norm of the difference of two unitary operations to a functional definition of similarity using the fidelity of states. 
\begin{definition}
Let $\ket{\psi}$ be a random variable defined on a distribution $\mathcal{J}=\text{Uni}[\partial B_{\mathbb{C}^N}(0,1)]$, i.e. $\ket{\psi}$ is uniformly random over all pure quantum states, we define two unitary operations $U_1$, $U_2$ to be pure-state \textbf{$(\delta, \epsilon)$-similar} if \begin{align*}
\mathbb{P}_{\psi\sim \mathcal{J}}(\mathcal{F}(U_1\ket{\psi}, U_2\ket{\psi})\geq 1-\epsilon)&\geq 1-\delta.
\end{align*} 
\end{definition}
Let $X_1$, $X_2$, ..., $X_m$ be independent random variables with {$X_i \in [a_i, b_i]\subset \mathbb{R}$} almost surely and define $S_m=\sum_{i=1}^m X_i$, the Chernoff-Hoeffding inequality \cite{hoeffding1994probability} states 
$$ \mathbb{P}(\abs{S_m-\mathbb{E}[S_m]}>\epsilon)<2e^{-\frac{2\epsilon^2}{\sum_{i=1}^m (b_i-a_i)^2}}.$$
A special case is when $X_1$, ... $X_m$ are iidrv on $[0,1]$ almost surely. Setting $\overline{X} = \frac{S_m}{m}$, we obtain the following inequality.
\begin{align*}
&\mathbb{P}(\abs{\overline{X}-\E[X_1]}>\epsilon)<2e^{-{2\epsilon^2 m}}.
\end{align*}
Note that when $m(\epsilon, \delta) =O( \frac{\log(2/\delta )}{2\epsilon^2})$, $\mathbb{P}(\abs{\overline{X}-\E[X_1]}\leq \epsilon)\geq 1-\delta$.
\section{Efficient Sampling Algorithms for Approximating Normalized Trace and Schatten Norms}
\label{sampling-ckt}
{
We first consider any matrix $A\in \mathbb{C}^{N\cross N}$ that is unitarily similar to a diagonal matrix (e.g. unitary matrices, Hermitian matrices, etc.) and present an efficient sampling technique to approximate its normalized trace. 
}
Recall that $A$ can be decomposed as: 
$$A=WDW^\dagger
=\sum_{i=1}^N d_i\ket{w_i}\bra{w_i}, w_i\in W$$
where $W$ is unitary and $D$ is diagonal.
The goal is to find a distribution $\mathcal{D}$ and a random vector $x\sim \mathcal{D}$ such that 
$$ \E_{x\sim  \mathcal{D}}\bra{x}A\ket{x}=\frac{\Tr(A)}{N}
.$$
{It suffices to show $\E_{x\sim \mathcal{D}}\bra{w}\ket{x} = \frac{1}{N}$ for all unit vectors $w\in \mathbb{C}^N$. As discussed in \cite{https://doi.org/10.48550/arxiv.0707.2831} and \cite{https://doi.org/10.48550/arxiv.1706.09279}, this is  equivalent to the uniform sampling of $\ket{x}$ (with replacement) from the standard basis $\{e_1,...,e_N\}$, which in general requires  $\Omega(N)$ sampling complexity. We show in Lemma \ref{lemma1} that if we construct $x(\theta)$ using a \textit{continuous} classical random variable $\theta$, then $\E_{\theta}\bra{w}\ket{x(\theta)} = \frac{1}{N}$ holds as desired. Under such a construction we can efficiently approximate the normalized trace $\frac{\Tr(A)}{N}$. 
}
\begin{lemma}
\label{lemma1}
Let $A\in \mathbb{C}^{N\cross N}$ be unitarily similar to a diagonal matrix and let $n=\lceil{\log N}\rceil$. Define a geometric sequence $(\omega)_{i=1}^{n}$ with $w_i=2^i$. Let random variable $\theta\sim \mathcal{D} = \text{Uni}[-\pi, \pi]$, and define a random vector $x(\theta) \in \mathbb{R}^N$ with $N$ entries $x_0, ..., x_{N-1}$ and $x_i (\theta) = \sqrt{\frac{2^n}{N}}\Pi_{j=1}^n \cos^{b_{i_j}}(\omega_j\theta)\sin^{1-b_{i_j}}(\omega_j\theta)$, where ${b_{i_1}...b_{i_n}}$ is the $n$-bit binary representation of $i$. For example, when $N=2^n-1$, we obtain $x(\theta)\in \mathbb{R}^N$ where the only missing entry is $\sin(\omega_1\theta)\sin(\omega_2\theta)...\sin(\omega_{n-1}\theta)\sin(\omega_{n}\theta)$,  \begin{align*}
&x(\theta)=\sqrt{\frac{2^n}{N}}\begin{pmatrix}
\cos(\omega_1\theta)\cos(\omega_2\theta)...\cos(\omega_{n-1}\theta)\cos(\omega_{n}\theta)\\
\cos(\omega_1\theta)\cos(\omega_2\theta)...\cos(\omega_{n-1}\theta)\sin(\omega_{n}\theta)\\
\cos(\omega_1\theta)\cos(\omega_2\theta)...\sin(\omega_{n-1}\theta)\cos(\omega_{n}\theta)\\
\cos(\omega_1\theta)\cos(\omega_2\theta)...\sin(\omega_{n-1}\theta)\sin(\omega_{n}\theta)\\
...\\
\sin(\omega_1\theta)\sin(\omega_2\theta)...\cos(\omega_{n-1}\theta)\cos(\omega_{n}\theta)\\
\sin(\omega_1\theta)\sin(\omega_2\theta)...\cos(\omega_{n-1}\theta)\sin(\omega_{n}\theta)\\
\sin(\omega_1\theta)\sin(\omega_2\theta)...\sin(\omega_{n-1}\theta)\cos(\omega_{n}\theta)\\
\end{pmatrix}.\end{align*}
Then
\begin{align*}
     &\E_\mathcal{\theta\sim D}\bra{x(\theta)}A\ket{x(\theta)}=\frac{\Tr(A)}{N}.
\end{align*}
\end{lemma}
\begin{proof}
{We first note that any signed sum of any subset of $\omega_i$'s is a \textit{nonzero} integer, i.e. $\forall S\subset \{1,...,n\}: \sum_{s\in S}\pm\omega_s\in \mathbb{Z}\backslash \{0\}$. This statement can be proved by induction, and we omit the proof here.}  For any such non-empty $S\subset \{1,...,n\}$, it follows that 
\begin{align}&\E_{\theta\sim \mathcal{D}}\Pi_{s\in S} e^{\pm i 2\omega_s\theta}=\frac{1}{2\pi} \int_{-\pi}^{\pi} e^{2\sum_{s\in S}\pm i\omega_s\theta}d\theta=0. \label{id}\end{align}
Next, we show that for all $j\neq k$, $\E_{\theta\sim \mathcal{D}}x_j x_k=0$ and $\E_{\theta\sim \mathcal{D}}{x_i}^2=\frac{1}{N}$.
For any pair $(j,k)$, 
let $c_{1}c_2...c_n$ be the binary representation of $j\oplus k$. We define $S_0=\{p\in \mathbb{N}\cap [1,n]: c_p=0\}$, $S_1=\{q\in \mathbb{N}\cap [1,n]:c_q=1\}$, and $2^{S_0}$, $2^{S_1}$ to be the corresponding power sets. Note that $S_0\cap S_1=\emptyset$. 
\begin{align*}
&\E_{\theta\sim \mathcal{D}} x_j x_k\\
&=\E_{\theta\sim \mathcal{D}} \frac{2^n}{N} \Pi_{p\in S_0}\frac{1\pm \cos(2\omega_p \theta)}{2}  \Pi_{q\in S_1}\frac{\sin(2\omega_q \theta)}{2}\\
&= \E_{\theta\sim \mathcal{D}}\frac{2^n}{N} \sum_{S\in 2^{S_0}}\frac{\pm 1}{2^{\abs{S_0}}}\frac{1}{2^\abs{S_1}} \Pi_{q\in S_1}{\sin(2\omega_q \theta)} \Pi_{p\in S} \cos(2\omega_p \theta)\\
&= \E_{\theta\sim \mathcal{D}}\frac{2^n}{N}\frac{1}{2^{n}}\sum_{S\in 2^{S_0}}\frac{\pm 1}{2^{\abs{S}}(2i)^\abs{S_1}} \Pi_{q\in S_1}(e^{i2\omega_q \theta} - e^{-i2\omega_q \theta})\Pi_{p\in S} (e^{i2\omega_p \theta} + e^{-i2\omega_p \theta}) \\
&=\frac{1}{N}\frac{1}{(2i)^\abs{S_1}}\sum_{S\in 2^{S_0}}\frac{\pm 1}{2^{\abs{S}}} \E_{\theta\sim \mathcal{D}} \Pi_{q\in S_1}(e^{i2\omega_q \theta} - e^{-i2\omega_q \theta})\Pi_{p\in S} (e^{i2\omega_p \theta} + e^{-i2\omega_p \theta}).
\end{align*}

When $j\neq k$, $j\oplus k\neq 0$ and $S_1\neq \emptyset$. It follows from (\ref{id}) that all $\E_{\theta\sim \mathcal{D}} \Pi_{q\in S_1}(e^{i2\omega_q \theta} - e^{-i2\omega_q \theta})\Pi_{p\in S} (e^{i2\omega_p \theta} + e^{-i2\omega_p \theta})$ evaluates to $0$. Therefore, $\E_{\theta\sim \mathcal{D}} x_j x_k=0$.

Analogously when $j=k$, $S_0=\{1,...,n\}$ and $S_1=\emptyset$. 
$$\E_{\theta\sim\mathcal{D}} x_j ^2=\frac{1}{N}+\frac{1}{N}\sum_{S\in 2^{S_0}\backslash \emptyset}\frac{\pm 1}{2^{\abs{S}}} \E_{\theta\sim \mathcal{D}} \Pi_{p\in S} (e^{i2\omega_p \theta} + e^{-i2\omega_p \theta}) = \frac{1}{N}.$$

We have thus shown that $x$ is unbiased under the standard basis $\{e_1, e_2, ..., e_N\}$ \cite{bandyopadhyay2002new}. {It remains to show} that for arbitrary unit vector $w$, $\E_{\theta\sim D}\abs{\bra{x}\ket{w_j}}^2=\frac{1}{N}$. We decompose $w$ in the standard basis, i.e. $w=\sum_{j=1}^N w_j e_j$.
$$ \E_{\theta\sim D}\abs{\bra{x}\ket{w}}^2
=\sum_{j,k=1}^N\E_{\theta\sim D}x_j x_k w^*_j w_k=\sum_{j=1}^N\E_{\theta\sim D}\abs{x_j}^2 \abs{w_j}^2
=\frac{1}{N}.
$$
The main claim $\E_\mathcal{\theta \sim D}\bra{x}A\ket{x}=\frac{\Tr(A)}{N}$ follows.

\end{proof}
{By Lemma \ref{lemma1}, }for an arbitrary matrix $A\in \mathbb{C}^{N\cross N}$ that's unitarily similar to a diagonal matrix, we can approximate $\frac{\Tr(A)}{N}$ using $m$ random samples. Namely, we randomly sample $\theta_1,...,\theta_m \sim \mathcal{D} = \text{Uni}[-\pi, \pi]$ and use $x(\theta_i)$ as defined in Lemma \ref{lemma1} to approximate \begin{align}  &\widehat{\frac{\Tr(A)}{N} }= \frac{1}{m}\sum_{i=1}^m \bra{x(\theta_i)}A \ket{x(\theta_i)}. ~\label{trace_ap}
\end{align}
{We study the sample complexity for such an approximation to achieve a low error rate with high success probability.}
\begin{theorem}
\label{theorem1}
Let $A\in \mathbb{C}^{N\cross N}$ be unitarily similar to a diagonal matrix. For any $\delta$, $\epsilon > 0$, with sample complexity $m(\epsilon, \delta)=O(\frac{\log(2/\delta)}{4\epsilon^2})$, samples $\theta_1,...,\theta_{m(\epsilon, \delta)} \sim \text{Uni}[-\pi, \pi]$, and $x(\theta_i)$ as defined in Lemma \ref{lemma1}, 
the following holds for the classical approximation of $\frac{\Tr(A)}{N}$ using (\ref{trace_ap}).
\begin{align*}
    &\mathbb{P}\Big(\abs{\widehat{\frac{\Tr(A)}{N} } -  \frac{\Tr(A)}{N}}< \epsilon \Big) > 1-\delta.
\end{align*}
\end{theorem}
\begin{proof}
The Theorem follows from the Chernoff-Hoeffding bound for complex numbers where we bound the precision for both real and imaginary parts to be within $\frac{\epsilon}{\sqrt{2}}$. 
\end{proof}
{The sampling method in (\ref{trace_ap})} can be generalized to estimate the normalized Schatten $p$-norms when $p$ is even, but in this paper we are particularly interested in the case when $p=2$. 

For arbitrary matrix $A\in \mathbb{C}^{N\cross M}$, $AA^\dagger$ is unitarily similar to a diagonal matrix in $\mathbb{C}^{N\cross N}$ with $\sigma_i^2$ along its diagonal. {We observe that $\norm{A}_{S_2}=\sqrt{\frac{\Tr(AA^\dagger)}{N}}$, based on which we approximate $\norm{A}_{S_2}$.} Namely, $$\widehat{\norm{A}}_{S_2} = \sqrt{\widehat{\frac{\Tr(AA^\dagger)}{N}}} \label{classical_s2}.$$

\begin{theorem}
\label{theorem2}
Let $A\in \mathbb{C}^{N\cross M}$ be an arbitrary matrix. For any $\delta$, $\epsilon > 0$, with sample complexity $m(\epsilon, \delta)=O(\frac{\log (2/\delta)}{2\epsilon^2}\min \{\epsilon^{-2}, \norm{A}_{S_2}^{-2}\})$, samples $\theta_1,...,\theta_{m(\epsilon, \delta)} \sim \text{Uni}[-\pi, \pi]$, and $x(\theta_i)$ as defined in Lemma \ref{lemma1}, the following holds for the classical approximation of $\norm{A}_{S_2}$ using (\ref{classical_s2}).
\begin{align*}
    &\mathbb{P}\Big(\abs{\widehat{\norm{A}}_{S_2} - \norm{A}_{S_2}}< \epsilon \Big) > 1-\delta.
\end{align*}
\end{theorem}
\begin{proof}
It suffices, see Appendix \ref{appendixthm2}, to show 
\begin{align*}
    &\mathbb{P}\Big(\abs{\widehat{\frac{\Tr(AA^\dagger)}{N}} - \norm{A}^2_{S_2}}< \epsilon \max\{\epsilon, \norm{A}_{S_2}\}\Big) > 1-\delta,
\end{align*}which follows from the result of normalized trace approximation in Theorem \ref{theorem1} and the Chernoff-Hoeffding bound.
\end{proof}

{Since all 
sampled state vectors $x(\theta_i)$ have dimension $N$}, classically evaluating each $\bra{x(\theta_i)}AA^\dagger \ket{x(\theta_i)}$ requires $\text{Poly}(N)$ arithmetic operations. Next, we show how we can take advantage of quantum computing to {
speed up the evaluations and thus make the approximation of the normalized Schatten 2-norms more efficient.}
\section{Quantum Approximation of Normalized Schatten $2$-norms for Mixed Quantum Operations}
\label{linear-comb}
{We apply the general results from the previous section to quantum operations. Let $n$ be the number of qubits in a quantum system and $N=2^n$. A quantum operation (usually denoted as $U$) can be defined by a unitary matrix in $\mathbb{C}^{N\times N}$. 
}

Recall that we define a mixed quantum operation to be a linear combination of a finite $K\sim O(1)$ quantum operations with coefficients $(\alpha)_{\kappa=1}^K$ satisfying $\sum_{\kappa =1}^K \abs{\alpha_{\kappa}}\leq 1$:
\begin{align*}
    &\tilde{U} = \sum_{\kappa=1}^K \alpha_{\kappa} U_{\kappa}.
\end{align*}

{While all $U_{\kappa}$ are unitary quantum operations which are unitarily similar to diagonal matrices, their linear combinations $\tilde{U}$ may not be.} Therefore, the approximation of the normalized trace from Theorem \ref{theorem1} does not {generalize 
to mixed quantum operations. Nevertheless, the approximation of the normalized Schatten $2$-norms does generalize, and we present explicit quantum circuit constructions to approximate $\tilde{\norm{{U}}}_{S_2}$.} We start with a toy example when $\tilde{U} = \frac{1}{\sqrt{2}}U_1-\frac{1}{\sqrt{2}}U_2$. For an arbitrary pure state $\ket{x}\in \partial B_{\mathbb{C}^N}(0,1)$, \begin{multline}
\bra{x}\tilde{U}\tilde{U}^\dagger\ket{x} 
=\frac{1}{2}\bra{x}(U_1-U_2)(U_1^\dagger-U_2^\dagger)\ket{x} 
=\frac{1}{2}(2 - 2\Re\{\bra{x}U_1U_2^\dagger\ket{x}\})  \\
=1-\Re\{\bra{x}U_1U_2^\dagger\ket{x}\}. \label{eq1}
\end{multline}
Similarly, we can generalize such result to mixed quantum operations $\tilde{U} = \sum_{\kappa=1}^K \alpha_{\kappa} U_{\kappa}$, where $\sum_{\kappa =1}^K \abs{\alpha_{\kappa}}\leq 1$.
\begin{align}
\bra{x}\tilde{U}\tilde{U}^\dagger\ket{x} &=\bra{x}\sum_{\kappa_1,\kappa_2=1}^K\alpha_{\kappa_1}\alpha^*_{\kappa2}U_{\kappa_1}U_{\kappa_2}^\dagger\ket{x}\nonumber\\
&=\sum_{\kappa=1}^K \abs{\alpha_{\kappa}}^2 + \sum_{1\leq \kappa_1<\kappa_2\leq K} 2\Re\{ \alpha_{\kappa_1}\alpha^*_{\kappa2}\bra{x}U_{\kappa_1}U_{\kappa_2}^\dagger\ket{x}\}\nonumber\\
&=\sum_{\kappa=1}^K \abs{\alpha_{\kappa}}^2 + \sum_{1\leq \kappa_1<\kappa_2\leq K} 2\Re\{ \alpha_{\kappa_1}\alpha^*_{\kappa2}\}\Re\{\bra{x}U_{\kappa_1}U_{\kappa_2}^\dagger\ket{x}\} \label{eq2}\\
&-\sum_{1\leq \kappa_1<\kappa_2\leq K} 2\Im\{ \alpha_{\kappa_1}\alpha^*_{\kappa2}\}\Im\{\bra{x}U_{\kappa_1}U_{\kappa_2}^\dagger\ket{x}\}\nonumber.
\end{align}
Since $U_{\kappa_1}, U_{\kappa_2}$ are quantum unitaries, the adjoint of them can be efficiently constructed 
by reversing the order of the gates. 

We show a simple construction of the quantum sampling circuit $S(\theta)$ (illustrated in Table \ref{U2}). {Let $n$ be the number of qubits, we define a geometric sequence $(\omega)_{i=1}^{n}$ with $w_i=2^i$.} \begin{align}
    &S(\theta) = \bigotimes_{i=1}^n R_y(2\omega_i \theta). \label{Sdef}
\end{align} 
Note that $S(\theta)\ket{\boldsymbol{0}}$ creates a quantum state $\ket{x(\theta)}$ {with a state vector matching the one defined} in Lemma \ref{lemma1}. 
\begin{table}[ht]
		\[\Qcircuit @C=1em  @R=1em {
			&\lstick{\ket{0}} & \gate{R_y(2\theta)}&\qw\\
			&\lstick{\ket{0}} & \gate{R_y(4\theta)}&\qw\\
			&\lstick{\ket{0}} & \gate{R_y(8\theta)}&\qw\\
		}\]

	 \caption{$S(\theta)$: sampling circuit for $3$ qubits.}\label{U2}	
\end{table}

{To apply the Hadamard test to evaluate $\Re\{ \alpha_{\kappa_1}\alpha^*_{\kappa2}\bra{x(\theta)}U_{\kappa_1}U_{\kappa_2}^\dagger\ket{x(\theta)}\}$, measurement circuits such as which in Table \ref{measurement_ckt} are used. }
\begin{table}[ht]
\[\Qcircuit @C=1em  @R=1em {
&\lstick{\ket{0}} & \gate{H} &\qw&\ctrl{1}&\ctrl{1}&\qw&\gate{H}&\qw&  \meter\\
&\lstick{\ket{\boldsymbol{0}}} & {/}\qw&\gate{S(\theta)}&\gate{U_{\kappa_2}^\dagger}&\gate{U_{\kappa_1}}
&\qw&\qw&\qw
}\]
\caption{Circuit for estimating $\Re\{\bra{\mathbf{0}}S^\dagger(\theta_i) U_{\kappa_1} U_{\kappa_2}^\dagger S(\theta_i)\ket{\mathbf{0}}\}$ using Hadamard test as in \ref{hadamard1}.}\label{measurement_ckt} 
\end{table}
\begin{lemma}
\label{lemma2}
$O(\frac{2\log(2/\delta)}{\epsilon^2})$ measurements suffice to bound the error from a Hadamard test to be within $\epsilon$ with probability $1-\delta$.
\end{lemma}
\begin{proof}
Assume {$m$ measurements are performed on the control bit of the Hadamard test with outcomes $M_1, M_2, ..., M_m \in \{0,1\}$,  we can approximate $\Pr(1)$ using $\widehat{\Pr(1)}=\frac{1}{m}\sum_{i=1}^m M_i$. Applying the Chernoff-Hoeffding bound,}
$$\mathbb{P} \Big( \abs{(1-2\widehat{\Pr(1)}) - (1-2\Pr(1))}>\epsilon \Big) =\mathbb{P} \Big( \abs{\widehat{\Pr(1)} - \Pr(1)}>\frac{\epsilon}{2} \Big) < 2e^{-{\epsilon^2}m/2}.$$
When $m\geq \frac{2\log(2/\delta)}{\epsilon^2}$, $2e^{-{\epsilon^2}m/2}\leq \delta$, and this completes the proof.
\end{proof}
\begin{lemma}
\label{lem-mea}
Given an arbitrary pure state $\ket{x}\in \partial B_{\mathbb{C}^N}(0,1)$ and a mixed quantum operation $\tilde{U} = \sum_{\kappa=1}^K \alpha_{\kappa} U_{\kappa}$ with $\sum_{\kappa =1}^K \abs{\alpha_{\kappa}}\leq 1$ and $K\sim O(1)$, $O(\frac{32K^4 \log(4K^2/\delta)}{\epsilon^2})$ measurements suffice to estimate $\bra{x}\tilde{U}\tilde{U}^\dagger\ket{x}$ to an error within $\epsilon$ with probability at least $1-\delta$.
\end{lemma}
\begin{proof}
According to equation (\ref{eq2}), $\bra{x}\tilde{U}\tilde{U}^\dagger\ket{x}$ can be estimated with a summation of $O(2K^2)$ measurements from the Hadamard tests. By the assumption $\sum_{\kappa =1}^K \abs{\alpha_{\kappa}}\leq 1$, $2\Re\{ \alpha_{\kappa_1}\alpha^*_{\kappa2}\}$ and $2\Im\{ \alpha_{\kappa_1}\alpha^*_{\kappa2}\}$ are both $\leq 2$. Applying the triangle inequality and the union bound, it suffices to bound the error of each Hadamard test to be within $\frac{\epsilon}{4K^2}$ with probability at least $1-\frac{\delta}{2K^2}$. We then apply Lemma \ref{lemma2} and obtain an upper bound on the sample complexity, $m=O(\frac{32K^4 \log(4K^2/\delta)}{\epsilon^2})$.
\end{proof}
With the assumption $K\sim O(1)$, Lemma \ref{lem-mea} implies a low-order polynomial measurement complexity to estimate the measurement outcome of the Hadamard test to a marginal error. Thus, we make an assumption that \textit{all measurements are error-free} from now on.

For an arbitrary mixed quantum operation $\tilde{U}$, we can approximate $\tilde{\norm{U}}_{S_2}$ using $m$ random samples. Namely, we randomly sample $\theta_1,...,\theta_m \sim \mathcal{D} = \text{Uni}[-\pi, \pi]$ and use the sampling circuit $S$ as defined by (\ref{Sdef}) to approximate \begin{align} \widehat{\tilde{\norm{U}}}_{S_2} = \sqrt{\frac{1}{m}\sum_{i=1}^m \bra{\boldsymbol{0}}S^\dagger(\theta_i)\tilde{U}\tilde{U}^\dagger S(\theta_i)\ket{\boldsymbol{0}}}\label{quantum-s2}
\end{align}
where each $\bra{\boldsymbol{0}}S^\dagger(\theta_i)\tilde{U}\tilde{U}^\dagger S(\theta_i)\ket{\boldsymbol{0}}$ can be measured with negligible error by Lemma \ref{lem-mea}.
\begin{theorem}
\label{theorem3} For an arbitrary mixed quantum operation $\tilde{U}$, we can estimate its normalized Schatten 2-norm efficiently using quantum sampling circuits of depth overhead $O(1)$. Moreover, for any $\delta$, $\epsilon > 0$, with sample complexity $m(\epsilon, \delta)=O(\frac{\log (2/\delta)}{2\epsilon^2}\min \{\epsilon^{-2}, \tilde{\norm{U}}_{S_2}^{-2}\})$, samples $\theta_1,...,\theta_{m(\epsilon, \delta)}\sim \text{Uni}[-\pi, \pi]$, and $S(\theta_i)$ as defined by (\ref{Sdef}), the following holds for the quantum approximation of $\tilde{\norm{{U}}}_{S_2}$ using (\ref{quantum-s2}).
$$
\mathbb{P}\Big(\abs{\widehat{\tilde{\norm{U}}}_{S_2} - \tilde{\norm{{U}}}_{S_2}}< \epsilon \Big) > 1-\delta. $$
\end{theorem}
\begin{proof}
The Theorem follows from Theorem \ref{theorem2}, Table \ref{measurement_ckt}, and Lemma \ref{lem-mea}.
\end{proof}
Note that the sample complexity for the approximation of the normalized Schatten 2-norm is independent of the number of qubits and is polynomial to $\frac{1}{\epsilon}$, which implies a potential quantum advantage. {In Figure \ref{thm3-fig}, we present simulated results supporting the relation $\epsilon \propto m(\epsilon, \delta)^{-1/2}$ which is in agreement with Theorem \ref{theorem3}.} 
\begin{figure}[h]
    \centering
    \includegraphics[width=8cm]{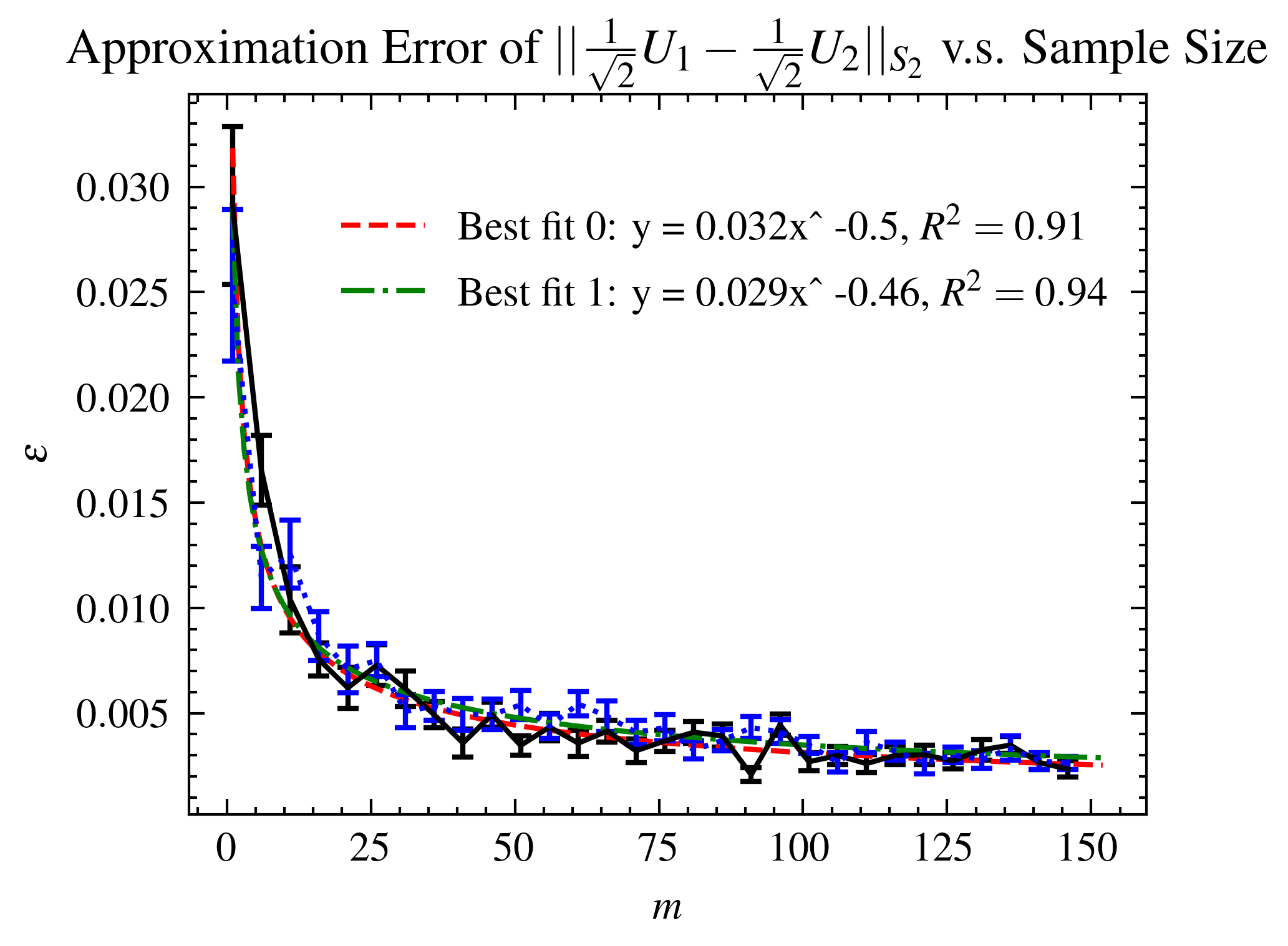}
    \caption{For each randomly generated mixed quantum operations $\tilde{U}=\frac{1}{\sqrt{2}}U_1-\frac{1}{\sqrt{2}}U_2$ where $U_1$, $U_2$ are randomly generated using the QR decomposition \cite{https://doi.org/10.48550/arxiv.math-ph/0609050}, we plot the error of the approximation $\widehat{\tilde{\norm{U}}}_{S_2}$ with respect to the sample size $m$ as defined in Theorem \ref{theorem3}. A 6-qubit system is considered. The means and standard errors are computed over 30 sets of random samples.}
    \label{thm3-fig}
\end{figure} 
In the next section, we build a connection from the normalized Schatten 2-norm of the difference of quantum operations to the similarity metric defined in Section \ref{background}.
\section{From the Normalized Schatten 2-Norm to a Fidelity-based Similarity Metric}
\label{similarity}
Recall the definition of pure-state $(\epsilon, \delta)$-similarity. Let $\ket{\psi}$ be a random state sampled from the distribution $\mathcal{J}=\text{Uni}[\partial B_{\mathbb{C}^N}(0,1)]$, we define two unitary operations $U_1$, $U_2$ to be pure-state $(\delta, \epsilon)$-similar if \begin{align*}
\mathbb{P}_{\psi\sim \mathcal{J}}(\mathcal{F}(U_1\ket{\psi}, U_2\ket{\psi})\geq 1-\epsilon)&\geq 1-\delta.
\end{align*} 
The following lemma is significant as it relates pure-state $(\epsilon, \delta)$-similarity to {the} normalized Schatten 2-norm.
\begin{lemma}\label{lemma3}
Let ${U}_1, U_2$ be two unitary quantum operations. ${U}_1$, ${U}_2$ are pure-state $(\epsilon, \delta)$-similar if $\norm{U_1-U_2}_{S_2} \leq \frac{\epsilon}{1+\sqrt{2(1/\delta - 1)}}$.
\end{lemma} 
\begin{proof}
We apply  statistical analysis to study the expectation and the variance of $\mathcal{F}(U_1\ket{\psi}, U_2\ket{\psi})$ when $\ket{\psi} \sim \mathcal{J}=\text{Uni}[\partial B_{\mathbb{C}^N}(0,1)]$. For any pure state $\ket{\psi}$,
\begin{multline*}
\mathcal{F}(U_1\ket{\psi}, U_2\ket{\psi})
=\abs{\bra{\psi}U_1^\dagger U_2\ket{\psi}}^2
\geq \Re^2\{\bra{\psi}U_1^\dagger U_2\ket{\psi}\}\\ =(1-\frac{\sum_{i=1}^N{\sigma_i}^2\abs{\bra{w_i}\ket{\psi}}^2}{2})^2 
\geq 1-\sum_{i=1}^N{\sigma_i}^2\abs{\bra{w_i}\ket{\psi}}^2,
\end{multline*}
where $w_i$ and $\sigma_i$ are the {left-singular} vectors and singular values of $U_1-U_2$. Let $\hat{\epsilon} = \norm{U_1-U_2}_{S_2}$ and $\ket{\psi} \sim \mathcal{J}$, 
$$\E_{\ket{\psi}\sim \mathcal{J}}\mathcal{F}(U_1\ket{\psi}, U_2\ket{\psi}) \geq 1-\norm{U_1-U_2}_{S_2}^2 \geq 1-\hat{\epsilon}^2.$$
{We next compute the variance of the fidelity},
\begin{align*}
\text{Var}_{\ket{\psi}\sim \mathcal{J}}[\mathcal{F}(U_1\ket{\psi}, U_2\ket{\psi})] &=\E_{\ket{\psi}} \mathcal{F}^2(U_1\ket{\psi}, U_2\ket{\psi}) - \big(\E_{\ket{\psi}} \mathcal{F}(U_1\ket{\psi}, U_2\ket{\psi})\big)^2\\&\leq 1 - (1-\hat{\epsilon}^2)^2 \leq 2\hat{\epsilon}^2.\end{align*}
{Application of the Chebyshev-Cantelli inequality,} for arbitrary $c>0$, 
$$\mathbb{P}_{\ket{\psi}\sim \mathcal{J}}\Big(\mathcal{F}(U_1\ket{\psi}, U_2\ket{\psi}) \geq 1-c\Big)\geq 1-\frac{2\hat{\epsilon}^2}{(c-\hat{\epsilon}^2)^2+2\hat{\epsilon}^2}.$$
Setting $c=\epsilon$ and $\frac{2\hat{\epsilon}^2}{(c-\hat{\epsilon}^2)^2+2\hat{\epsilon}^2} = \delta$, it suffices to have $\hat{\epsilon}\leq \frac{\epsilon}{1+\sqrt{2(1/\delta - 1)}}$.
\end{proof}
{Empirical relations between $\mathcal{F}(U_1\ket{\psi}, U_2\ket{\psi})$ and $\norm{U_1-U_2}_{S_2}$ are illustrated in Figures \ref{fig-sim1} and \ref{fig-sim2}. In both experiments, $U_1$ is a fixed random unitary generated by QR decomposition \cite{https://doi.org/10.48550/arxiv.math-ph/0609050} and $U_2$ is constructed by applying rotation operators to $U_1$. Figure \ref{fig-sim2} supports the bound derived in Lemma \ref{lemma3} as the probability for $U_1$, $U_2$ to be pure-state $\big((1+\sqrt{8})\norm{U_1-U_2}_{S_2}, 0.2\big)$-similar is much higher than $0.8$ for all pairs of $U_1$ and $U_2$ used (setting $\delta=0.2$).}
\begin{figure}[h]
    \centering
    \includegraphics[width=8cm]{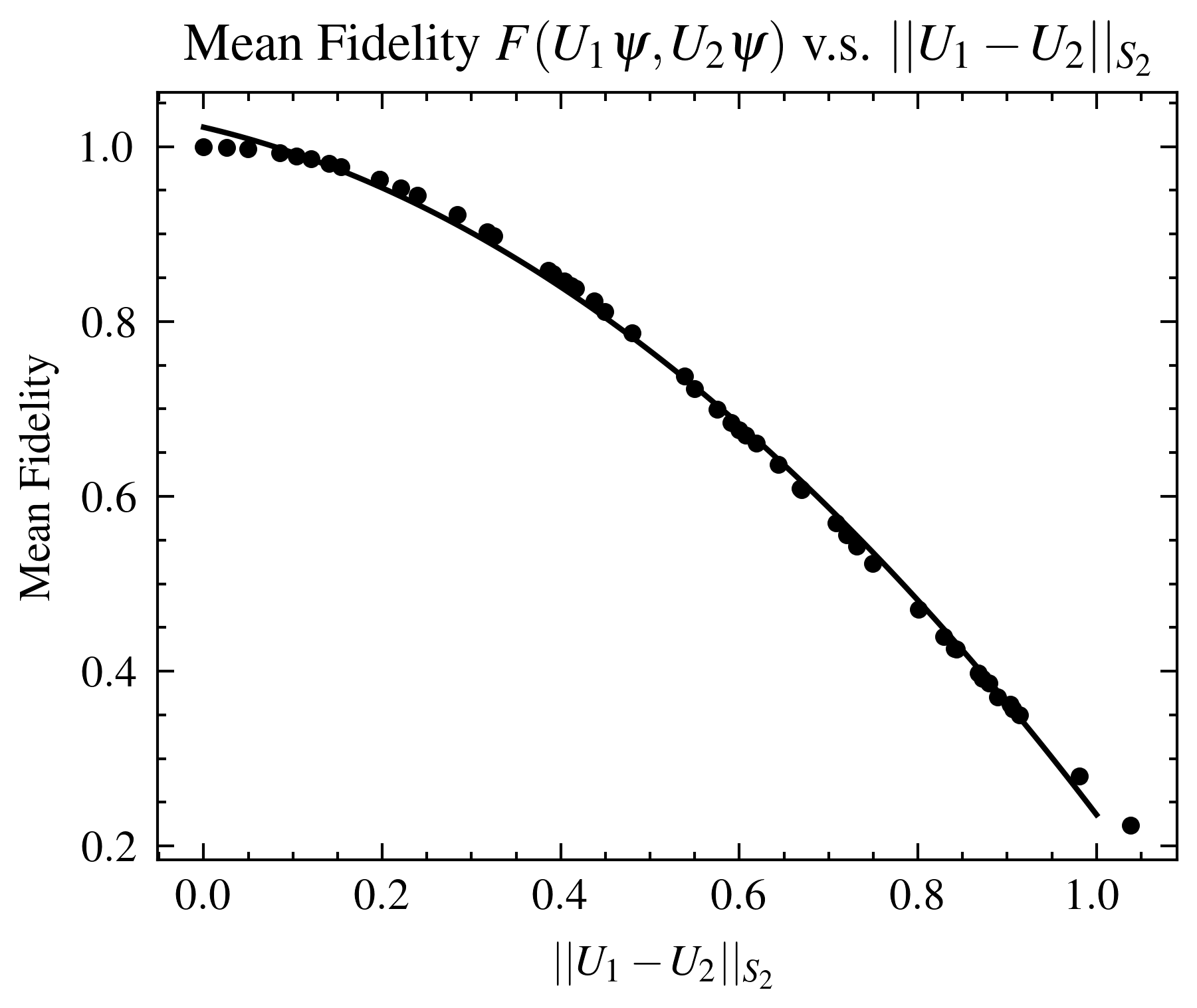}
    \caption{Empirical relation between approximated $\E_{\psi\sim \text{Uni}[\partial B_{\mathbb{C}^N}(0,1)]}\mathcal{F}(U_1\ket{\psi}, U_2\ket{\psi})$ and $\norm{U_1-U_2}_{S_2}$. For each pair of $U_1$ and $U_2$, the mean fidelity is computed over 1000 randomly sampled (with replacement) states $\psi\sim \text{Uni}[\partial B_{\mathbb{C}^N}(0,1)]$. A 6-qubit system is considered.}
    \label{fig-sim1}
\end{figure}
\begin{figure}[h]
    \centering
    \includegraphics[width=8cm]{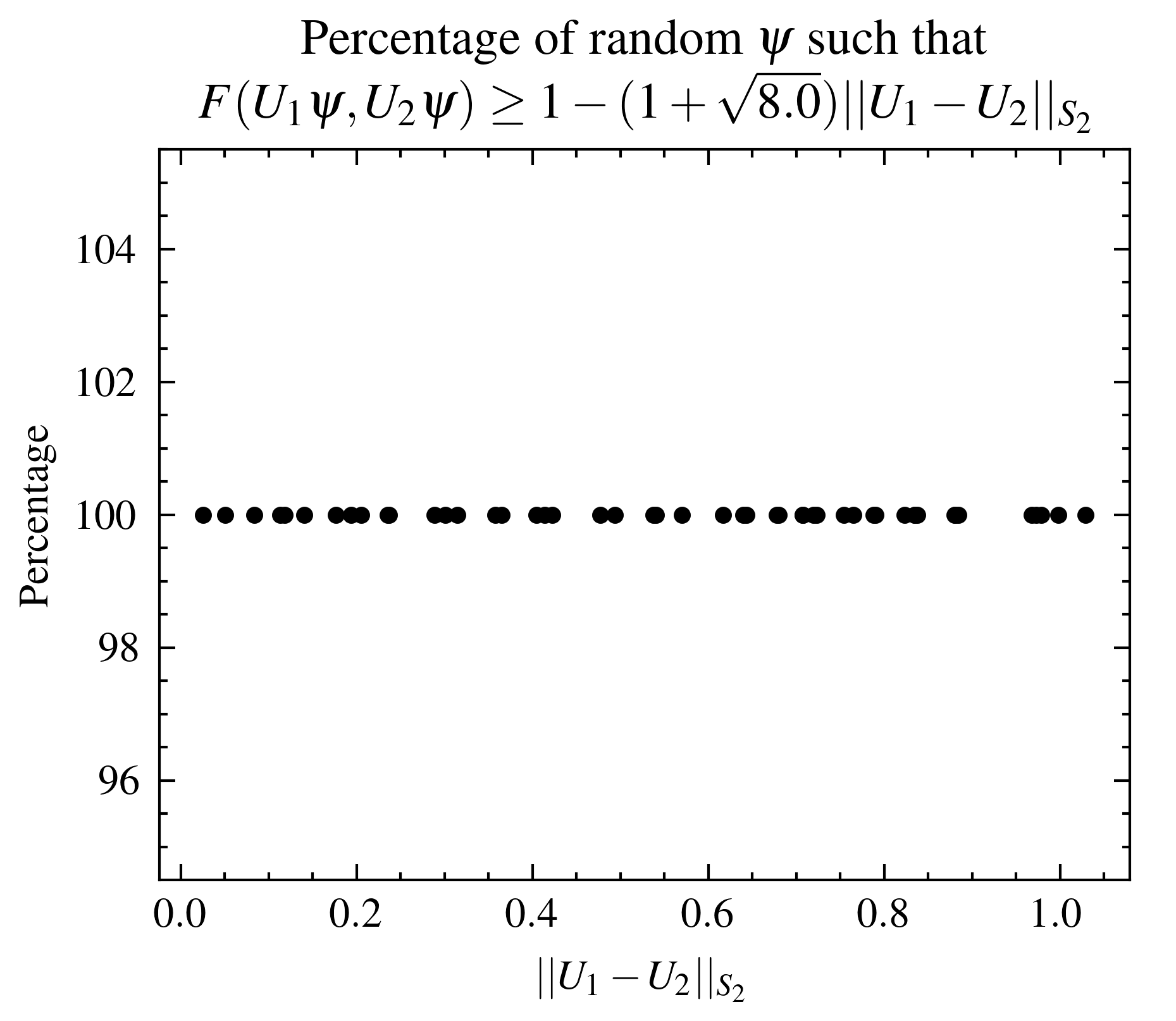}
    \caption{Verification of Lemma \ref{lemma3} when $\delta=0.2$. For each pair of $U_1$ and $U_2$, the percentage is computed over 1000 randomly sampled (with replacement) $\psi\sim \text{Uni}[\partial B_{\mathbb{C}^N}(0,1)]$. A 6-qubit system is considered.}
    \label{fig-sim2}
\end{figure}

{The lemma can be generalized to mixed quantum operations. Let $\mathcal{J}=\text{Uni}[\partial B_{\mathbb{C}^N}(0,1)]$, we define two mixed quantum operations $\tilde{U}_1, \tilde{U}_2$ to be pure-state $(\epsilon,\delta)$-similar if }
$$\mathbb{P}_{\ket{\psi}\sim \mathcal{J}}\Big(\mathcal{F}(\tilde{U}_1\ket{\psi}, \tilde{U}_2\ket{\psi})\geq \E_\psi ^2\frac{\bra{\psi}\tilde{U}_1\tilde{U}_1^\dagger + \tilde{U}_2\tilde{U}_2^\dagger\ket{\psi}}{2}-\epsilon\Big)\geq 1-\delta.$$
\begin{lemma}
\label{lemma5}
Let $\tilde{U}_1, \tilde{U}_2$ be two mixed quantum operations and $\mathcal{J}=\text{Uni}[\partial B_{\mathbb{C}^N}(0,1)]$. $\tilde{U}_1$, $\tilde{U}_2$ are pure-state $(\epsilon, \delta)$-similar if $\norm{\tilde{U}_1-\tilde{U}_2}_{S_2} \leq \sqrt{\frac{\epsilon^2-(1/\delta - 1)(\tau-\tau^4)}{2\tau(\epsilon +(1/\delta-1)\tau^2)}}$, where $\tau = \E_{\ket{\psi} \sim \mathcal{J}} \frac{\bra{\psi}\tilde{U}_1\tilde{U}_1^\dagger + \tilde{U}_2\tilde{U}_2^\dagger\ket{\psi}}{2}$.
\end{lemma}
\begin{proof}
For any given pure state $\ket{\psi}\in \partial B_{\mathbb{C}^N}(0,1)$,
\begin{align*}
\mathcal{F}(\tilde{U}_1\ket{\psi}, \tilde{U}_2 \ket{\psi}) & \geq \Re^2\{\bra{\psi}\tilde{U}_1^\dagger \tilde{U}_2\ket{\psi}\}\\&=(\frac{\bra{\psi}\tilde{U}_1\tilde{U}_1^\dagger\ket{\psi} + \bra{\psi}\tilde{U}_2\tilde{U}_2^\dagger\ket{\psi}}{2}-\frac{\sum_{i=1}^N{\sigma_i}^2\abs{\bra{w_i}\ket{\psi}}^2}{2})^2,
\end{align*}
where $w_i$ and $\sigma_i$ are the {left-singular} vectors and singular values of $\tilde{U}_1-\tilde{U}_2$.
Let $\tau = \E_{\ket{\psi} \sim \mathcal{J}} \frac{\bra{\psi}\tilde{U}_1\tilde{U}_1^\dagger\ket{\psi} + \bra{\psi}\tilde{U}_2\tilde{U}_2^\dagger\ket{\psi}}{2}$ and $\hat{\epsilon} = \norm{\tilde{U}_1-\tilde{U}_2}_{S_2}$,
$$ \E_{\ket{\psi} \sim \mathcal{J}} \mathcal{F}(\tilde{U}_1\ket{\psi}, \tilde{U}_2 \ket{\psi})
\geq \E_{\ket{\psi}} \Big(\frac{\bra{\psi}\tilde{U}_1\tilde{U}_1^\dagger\ket{\psi} + \bra{\psi}\tilde{U}_2\tilde{U}_2^\dagger\ket{\psi}}{2}\Big)^2-\tau\hat{\epsilon}^2 \geq \tau^2-\tau\hat{\epsilon}^2.$$
\begin{align*}
\text{Var}_{\ket{\psi} \sim \mathcal{J}}[\mathcal{F}(\tilde{U}_1\ket{\psi}, \tilde{U}_2\ket{\psi})] &=\E_{\ket{\psi}} \mathcal{F}^2(\tilde{U}_1\ket{\psi}, \tilde{U}_2\ket{\psi}) - \big(\E_{\ket{\psi} } \mathcal{F}(\tilde{U}_1\ket{\psi}, \tilde{U}_2\ket{\psi})\big)^2\\&\leq \tau - (\tau^2-\tau\hat{\epsilon}^2)^2 \leq \tau-\tau^4+2\tau^3\hat{\epsilon}^2.
\end{align*}
Applying the Chebyshev-Cantelli inequality, for arbitrary $c>0$, 
$$\mathbb{P}_{\ket{\psi} \sim \mathcal{J}}\Big(\mathcal{F}(\tilde{U}_1\ket{\psi}, \tilde{U}_2\ket{\psi}) \geq \tau^2-c\Big)\geq 1-\frac{\tau-\tau^4+2\tau^3\hat{\epsilon}^2}{(c-\tau\hat{\epsilon}^2)^2+\tau-\tau^4+2\tau^3\hat{\epsilon}^2}.$$
Setting $c=\epsilon$ and $\frac{\tau-\tau^4+2\tau^3\hat{\epsilon}^2}{(c-\tau\hat{\epsilon}^2)^2+\tau-\tau^4+2\tau^3\hat{\epsilon}^2} = \delta$, it suffices to have $\hat{\epsilon}\leq \sqrt{\frac{\epsilon^2-(1/\delta - 1)(\tau-\tau^4)}{2\tau(\epsilon +(1/\delta-1)\tau^2)}}$.
\end{proof}

{Unlike for unitary quantum operations, the error bound for mixed quantum operations depends on $\tau$, which incorporates the difference of $\tilde{U}_1$ and $\tilde{U}_2$ in "size." Combining Lemma \ref{lemma3} and Theorem \ref{theorem3}, we obtain the following theorem.}
\begin{theorem}
\label{theorem4} Consider arbitrary $\epsilon, \delta, \hat{\delta}>0$ and input unitary quantum operations ${U}_1$, ${U}_2$. 
Consider $m$ independent samples $\theta_1, ..., \theta_m \sim \text{Uni}[-\pi, \pi]$. Let $S(\theta_i)$ be as defined by (\ref{Sdef}) and the corresponding quantum approximations be as defined by (\ref{quantum-s2}). Then $\mathbb{P}(U_1,U_2 \:\text{are}\: (\epsilon, \delta)\text{-similar})\geq 1-\hat{\delta}$ if the following inequality holds.
\begin{align*}
 &\widehat{{\norm{{U}_1-U_2}}}_{S_2} + \min \{\sqrt[4]{{\frac{2\log(2/\hat{\delta})}{m}}},\sqrt{\frac{2\log(2/\hat{\delta})}{m}}\norm{U_1-U_2}^{-1}_{S_2}\}
 \leq \frac{\epsilon}{1+\sqrt{2(1/\delta - 1)}}.
\end{align*} 
\end{theorem}
\begin{proof}
Following from Lemma \ref{lemma3}, it suffices to show that $\norm{U_1-U_2}_{S_2}\leq \frac{\epsilon}{1+\sqrt{2(1/\delta - 1)}}$ with probability at least $1-\hat{\delta}$. By Theorem \ref{theorem3}, with $m$ samples, \begin{align*}
    &\mathbb{P}\Big(\frac{\abs{\widehat{{\norm{{U}_1-U_2}}}_{S_2}-\norm{U_1-U_2}_{S_2}}}{\sqrt{2}}\leq  
    \min \{\sqrt[4]{{\frac{\log(2/\hat{\delta})}{2m}}},\sqrt{\frac{\log(2/\hat{\delta})}{2m}}\sqrt{2}\norm{U_1-U_2}^{-1}_{S_2}\} \Big)
    \\ &\geq 1-\hat{\delta}.
\end{align*}
After simplifying and loosening the bound a little, we obtain
\begin{align*}
&\mathbb{P}\Big({{\norm{U_1-U_2}_{S_2}}}\leq \widehat{{\norm{{U}_1-U_2}}}_{S_2}+ \min \{\sqrt[4]{{\frac{2\log(2/\hat{\delta})}{m}}},\sqrt{\frac{2\log(2/\hat{\delta})}{m}}\norm{U_1-U_2}^{-1}_{S_2}\}\Big)\\&\geq 1-\hat{\delta}.
\end{align*}
The main claim follows.
\end{proof}

\section{Applications to Quantum Circuit Learning}

\label{learning}
Quantum circuit learning is one of the most natural applications of the similarity metric. A problem setting is  as follows: {given a target unitary quantum operation $V$,   represented via its clean-qubit controlled circuit},  find a  parameter set $\hat{\xi}$ of a variational circuit $U(\xi)$ that best approximates $V$. {Theorem \ref{theorem4} inspires a circuit learning algorithm whose cost function, $\widehat{{\norm{{U(\xi)}-V}}}_{S_2}^2$, utilizes the normalized Schatten 2-norm of the difference between $U(\xi)$ and $V$ (see Algorithm \ref{algorithm}). We increase the similarity between $U(\xi)$ and $V$ by minimizing the cost function.}
\begin{algorithm}[ht]
\caption{Sample based quantum circuit learning}\label{algorithm}
\textbf{Input}: Target unitary $V$. \\
\textbf{Output}: $\mathbf{\xi}$ and $ U(\mathbf{\xi})$ which approximates $V$.
\begin{algorithmic}[1]
\State $\Theta=\emptyset$ \Comment{$\Theta$ saves the random samples}
\For{$i: 1\to m$} \Comment{Generating $m$ samples}
\State $\theta_i\sim \text{Uniform}[-\pi, \pi]$ \Comment{$\mathcal{D}=\text{Uniform}[-\pi, \pi]$}
\State $\Theta=\Theta\cup \theta_i$
\EndFor
\State Randomly initialize $\mathbf{\xi}$.

\Loop
	\State $f(\mathbf{\xi}) = 2 - \frac{1}{m}\sum_{i=1}^m2\Re\{\bra{\mathbf{0}}S^\dagger (\theta_i)V^\dagger U(\mathbf{\xi})S(\theta_i)\ket{\mathbf{0}}\}$ \Comment{Objective}
	\State $\mathbf{\xi} = \mathbf{\xi} - \eta \nabla_{\mathbf{\xi}} f(\mathbf{\xi})$ \Comment{Gradient Descent}
\EndLoop
\end{algorithmic}
\end{algorithm}
A quantum circuit diagram for approximating $\Re\{\bra{\mathbf{0}}S^\dagger (\theta_i)V^\dagger U(\mathbf{\xi})S(\theta_i)\ket{\mathbf{0}}\}$ is shown in Table \ref{learningckt}.
\begin{table}[ht]
\[\Qcircuit @C=1em  @R=1em {
&\lstick{\ket{0}} & \gate{H} &\qw&\ctrl{1}&\ctrl{1}&\qw&\gate{H}&\qw&  \meter\\
&\lstick{\ket{\boldsymbol{0}}} & {/}\qw&\gate{S(\theta)}&\gate{U(\mathbf{\xi})}&\gate{V^\dagger}&\qw
&\qw&\qw
}\]
\caption{Circuit for $3$ qubits: approximating $\Re\{\bra{\mathbf{0}}S^\dagger (\theta_i)V^\dagger U(\mathbf{\xi})S(\theta_i)\ket{\mathbf{0}}\}$}\label{learningckt}
\end{table}
To obtain an estimate of the gradient with respect to $\xi_i$, one could apply a black-box gradient approximation \cite{https://doi.org/10.48550/arxiv.1909.05074,Stokes_2020,Harrow_2021,Mitarai_2018,parameter_shift}. 
\begin{align*}
&\frac{\partial f}{\partial \xi_i} \approx_{\varepsilon\to 0} \frac{f(\xi_i+\varepsilon)-f(\xi_i-\varepsilon)}{2\varepsilon}.
\end{align*}
A better accuracy of the approximation can be achieved by increasing $m$, which is consistent with Theorem \ref{theorem4}. One possible application of the algorithm is learning the square root of a quantum operation $V$, where we use $U(\xi)U(\xi)$ to approximate $V$. 

\section{Concluding Remarks}
In summary, we defined and introduced the normalized Schatten norms and a set of similarity metrics between quantum operations that can be efficiently estimated. We discussed sufficient and necessary conditions for a sampling circuit to estimate the normalized Schatten 2-norm and showed one optimal design of such sampling circuits. We then studied the sample complexity required by the sampling circuit and obtained an upper bound that is polynomial to $\frac{1}{\epsilon}$. With such an efficient sampling method, we were able to estimate the normalized Schatten 2-norms of mixed quantum operations.  We next related similarity of quantum operations based on the normalized Schatten 2-norm to a similarity metric induced by the traditional fidelity metric used for quantum states. Finally, we showed how such a connection could lead to a design of the loss function for circuit learning for tasks such as approximating a given quantum circuit or its square root.

In this paper we emphasized circuit learning applications to the problem of approximating unitary operations. We did not explore learning of mixed operations. A similar circuit learning approach would apply to mixed quantum operations with a corresponding modification of the loss term at the line 8 of Algorithm \ref{algorithm}. However, as noted in Lemma \ref{lemma5}, the error bound posed on the normalized Schatten 2-norm of the difference of the mixed operations is also a function of $\tau =\E_{\ket{\psi} \sim \mathcal{J}} \frac{\bra{\psi}\tilde{U}_1\tilde{U}_1^\dagger + \tilde{U}_2\tilde{U}_2^\dagger\ket{\psi}}{2}$. When $\tau$ is not close 1, there's a weaker correlation between the normalized Schatten 2-norm of the difference and the fidelity-based similarity metric. 
\bibliography{sn-bibliography}
\newpage
\appendix
\section{Theorem \ref{theorem2} Supplement}
\label{appendixthm2}
\begin{align*}
&\mathbb{P}\bigl( \biggl\vert \widehat{\frac{\Tr(AA^\dagger)}{N}} - \norm{A}^2_{S_2} \biggr\vert < \epsilon \max\{\epsilon, \norm{A}_{S_2}\}\bigr) > 1-\delta 
\\
&\implies\\
&\mathbb{P}\Big(\abs{\widehat{\norm{A}}_{S_2} - \norm{A}_{S_2}}< \epsilon \Big) > 1-\delta.
\end{align*}
\begin{proof}
Based on our classical approximation algorithm defined by (\ref{classical_s2}), 
\begin{align*}
\widehat{\norm{A}}_{S_2}=\sqrt{\widehat{\frac{\Tr(AA^\dagger)}{N}}}.
\end{align*} Let $\mathcal{M} =\widehat{\frac{\Tr(AA^\dagger)}{N}}= \widehat{\norm{A}}_{S_2}^2$, we divide the proof into two cases.\\
If $\epsilon < \norm{A}_{S_2}$,
\begin{align*}
\abs{\mathcal{M}-{\norm{A}^2_{S_2}}}\leq \epsilon\norm{A}_{S_2} &\implies {\norm{A}^2_{S_2}}-\epsilon\norm{A}_{S_2}\leq \mathcal{M}\leq {\norm{A}^2_{S_2}}+\epsilon\norm{A}_{S_2}\\
&\implies {\norm{A}^2_{S_2}}-2\epsilon\norm{A}_{S_2}+\epsilon^2\leq {\mathcal{M}}\leq {\norm{A}^2_{S_2}}+2\epsilon\norm{A}_{S_2}+\epsilon^2\\
&\implies \big(\norm{A}_{S_2}-\epsilon\big)^2\leq {\mathcal{M}}\leq\big(\norm{A}_{S_2}+\epsilon\big)^2\\
&\implies \abs{\sqrt{\mathcal{M}}-\norm{A}_{S_2}} \leq \epsilon.
\end{align*}
If $\epsilon \geq \norm{A}_{S_2}$,
\begin{align*}
\abs{{\mathcal{M}}-{\norm{A}^2_{S_2}}} \leq \epsilon^2 &\implies {\norm{A}^2_{S_2}}-\epsilon^2\leq \mathcal{M}\leq {\norm{A}^2_{S_2}}+\epsilon^2\\
&\implies 0\leq \mathcal{M}\leq {\norm{A}^2_{S_2}}+2\epsilon\norm{A}_{S_2}+\epsilon^2\\
&\implies 0\leq \mathcal{M}\leq\big(\norm{A}_{S_2}+\epsilon\big)^2\\
&\implies \abs{\sqrt{\mathcal{M}}-\norm{A}_{S_2}} \leq \epsilon.
\end{align*}
\end{proof}
\end{document}